\documentclass{article}
\usepackage[colorlinks=true, linkcolor=blue, citecolor=blue,urlcolor=black]{hyperref}
\usepackage[utf8]{inputenc} 
\usepackage[T1]{fontenc}    
\usepackage{hyperref}       
\usepackage{url}            
\usepackage{booktabs}       
\usepackage{amsfonts}       
\usepackage{nicefrac}       
\usepackage{microtype}      
\usepackage{xcolor}         
\usepackage{graphicx}
\usepackage{forest}
\usepackage{algorithm}
\usepackage[noend]{algpseudocode}
\usepackage[utf8]{inputenc}
\usepackage{comment}
\usepackage{amsmath}
\usepackage{amssymb}
\usepackage{amsthm}
\usepackage{geometry}
\usepackage{graphicx}
\graphicspath{ {./images/} }
\usepackage{amsmath,amssymb}
\usepackage{fullpage}
\usepackage{color,soul}   
\usepackage{hyperref}
\usepackage{todonotes}
\hypersetup{
	colorlinks=true, 
	linktoc=all,     
	linkcolor=blue,  
}

\newtheorem{theorem}{Theorem}

\newtheorem{remark}{Remark}
\newtheorem{claim}{Claim}
\newtheorem{lemma}{Lemma}
\newtheorem{corollary}{Corollary}

\usepackage{mathtools}

\graphicspath{{./},{./figures/}}
\def\mathclap#1{\text{\hbox to 0pt{\hss$\mathsurround=0pt#1$\hss}}}

\title{Bode Integral Limitation For Irrational Systems}

\author{William Chang,\thanks{W.\ Chang is with the Department of Mathematics, University of Southern California, Los Angeles, CA; chan087@usc.edu.}\and Fariba Ariaei \thanks{F. Ariaei is with the Department of Mechanical and Aerospace Engineering, University of California in Irvine, Irvine, CA; fariaei@uci.edu.} \and Edmond Jonckheere \thanks{E.\ Jonckheere is with the Department of Electrical Engineering, University of Southern California, Los Angeles, CA; jonckhee@usc.edu.}}
\date{}
\begin{document}
\maketitle

\begin{abstract}
Bode integrals of sensitivity and sensitivity-like functions along with complementary sensitivity and complementary sensitivity-like functions are conventionally used for describing performance limitations of a feedback control system. In this paper, we investigate the Bode integral and evaluate what happens when a fractional order Proportional-Integral-Derivative (PID) controller is used in a feedback control system. We extend our analysis to when fractal PID controllers are applied to irrational systems. We split this into two cases: when the sequence of infinitely many right half plane open-loop poles doesn't have any limit points and when it does have a limit point. In both cases, we prove that the structure of the Bode Integral is similar to the classical version under certain conditions of convergence. We also provide a sufficient condition for the controller to lower the Bode sensitivity integral. 
\end{abstract}

\section{Introduction}

Control systems design is contingent on many performance considerations and physical limitations, and is seen invariably a tradeoff between the quest for high performance goals and the need for meeting hard design constraints. However, not all design goals are achievable and the fundamental question is what system characteristics may impose inherent constraints on design and implementation.

Systematic investigation and understanding of fundamental control limitations date back to the classical work of Bode in the 1940s on the logarithmic sensitivity integral, known as the \textit{Bode sensitivity integral} relation (\cite{Bode1945}). Bode’s work has had a lasting impact on the theory and practice of control, and has inspired continued research efforts ever since, resulting in a variety of extensions and new results that seek to quantify design constraints and performance limitations by logarithmic integrals of Bode and Poisson type (see, e.g., \cite{freudenberg1985right}, \cite{BodeSensitivity2022}). In feedback control, Bode's Sensitivity integral quantifies the restriction on sensitivity function $S(s)$, defined as $S(s):=\frac{1}{1+L(s)}$, the transfer function between the reference input to the tracking error or an output disturbance signal to the output, where $L(s)$ is the loop transfer function (see Fig. \ref{fig:unityFeedback}). The theorem developed by Hendrik Wade Bode states
\begin{figure}
    \centering
    \includegraphics[width = .3\textwidth]{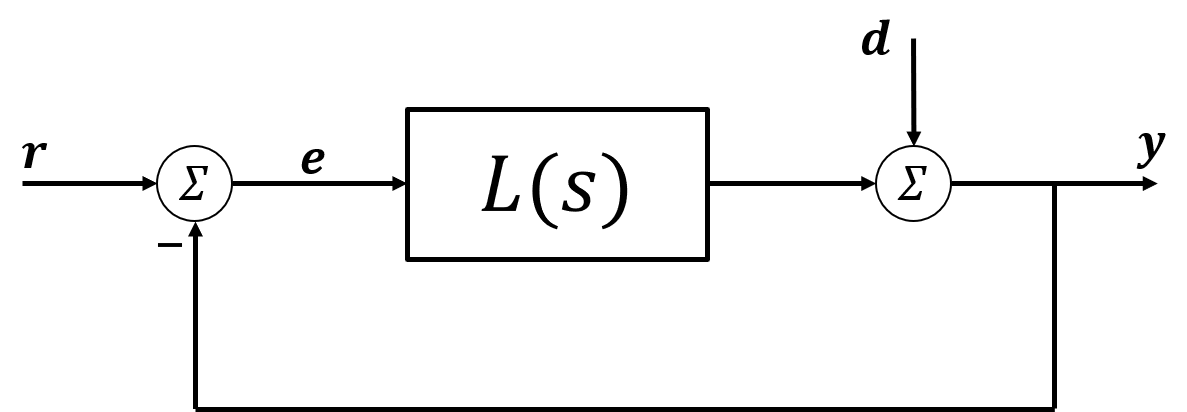}
    \caption{Unity feedback system.}
    \label{fig:unityFeedback}
\end{figure}
\begin{align} \label{eq:Bode}
\int_{-\infty}^\infty \ln|S(iw)|^2 dw = 4\pi\sum_k \mathfrak{R}(p_k) - 2\pi\lim_{s \to \infty} s L(s),
\end{align}
where $p_k$'s are the poles of $L(s)$ in the right half plane.
There is extensive literature on sensitivity of control systems and the fundamental and inevitable design limitations for linear time-invariant (LTI) systems. In the search for the best achievable performance, the performance limitation of feedback systems has also been extensively investigated in the study of optimal control problems, leading to the discovery of fundamental performance limits defined under various criteria. Most notably, this area of study has witnessed in recent years significant theoretical developments in multivariable, sampled-data, time-varying, nonlinear, networked control, communication, and applications to technological fields including, e.g., Active Queue Management (AQM) and TCP congestion control (\cite{shah2004performance}), acoustic control (\cite{pang2005suppressing}), motion control (\cite{subramanian2018discrete,sariyildiz2021guide,aangenent2005nonlinear}), and biological systems (\cite{BodeInBiology}). 
There has also been extensive studies on information theoretic analysis of dynamical systems. \cite{jonckheere1992chaotic} show that Kolmogorov-Sinai entropy cannot be reduced by linear, stationary feedback if open loop system is stable, and in general, the Shannon entropy rate can't be decreased because of the Bode limitation. \cite{zhang2002information,zhang2003bode} use Bode integral to connect entropy rate of system output with strictly unstable poles of open-loop transfer function, and to formulate system performance limits in the framework of information theory. Later in \cite{li2012bode, wan2018sensitivityCDC, wan2019sensitivity}, and \cite{roy2011studies}, mutual information is used to analyze the Bode sensitivity integral and performance limitations in communication-control systems, when stochastic disturbances and noise are present.
The interest in Bode’s work has also recently led to several extensions
to the aforementioned Bode’s theorem. \cite{chen1997sensitivity} provides derivations for Bode and Poisson-type integral relations for both continuous-time and discrete-time systems based on properties of Laplace and Z transformations, respectively. Whereas, \cite{wu1992simplified} and \cite{wan2019simplified,wan2022simplified}, used a simplified approach to analyze Bode and Bode-type filtering sensitivity integrals for both continuous-time and discrete-time systems without using Cauchy integral theorem or Poisson integral formula. In a similar approach, \cite{emami2019bodeRevisited} show that the sensitivity integral constraint is crucially related to the difference in speed (bandwidth) of the closed-loop and the open-loop system. The analysis and results of \cite{emami2019bodeRevisited}, \cite{freudenberg1985right, Freudenberg87, Freudenberg1988FrequencyDP}, and  \cite{Kwakernaak1972} are restricted to systems with rational loop transfer functions and specific constraints on the characteristics of systems at high frequencies. 

\textbf{Our Contributions:} We develop a specialized Bode integral limitation for irrational systems with fractal PID controllers. Irrational systems have fractional-order models and are more adequate than the previously used integer-order models, especially for the description of long memory and hereditary effects in various physical media, as well as for modeling dynamical processes in fractal media (as defined by \cite{Mandelbrot1982}). However, fractional-order dynamic systems have been studied marginally with regard to fundamental limitations. In this paper, we will show the advantages of using fractal PID controllers against their integer counterparts, namely PID controllers. We use complex variable theory to prove the Bode sensitivity integral for fractional systems and show that the Bode sensitivity integral can be reduced. The results are significant as they indicate with a suitable choice of fractional orders, the fractal PID controllers provide a better performance regardless of robustness conditions. Possible tradeoffs will be examined in future work.

The remainder of this paper is organized as follows. In section \ref{sec:finite} we use the calculus of residues and contour analysis to evaluate the bode integral. In this section, we assume the system's model is rational. In section \ref{sec:infinite}, we extend this analysis to general sensitivity functions. We apply Weierstrass factorization to $S$ to write it as a quotient of two polynomial signals where the degree of the polynomials is allowed to be infinite. In particular, we focus our analysis on when there are infinitely many open right half plane poles. We split  this into two cases: when the sequence of poles don't have any limit points and when they do have limit point. In both cases, we find that the form of the Bode integral is similar to the classical version under certain conditions of convergence. In our analysis, we assume the contour integral around the large semicircle goes to $0$ as the radius of the semi-circle gets larger; however other assumptions can easily make our results more generalizable. Finally, we describe future work in section \ref{sec:conclusion}.

\section{Finite open right half plane poles}\label{sec:finite}

Let $I(S)$ be the Bode integral of the sensitivity function $S$ (defined in \eqref{eq:Bode}) of a control system in Fig. \ref{fig:controlSys}. For our analysis, we assume that $S(s)$ is a general function $S: \mathbb{C} \to \mathbb{C}$. The poles of $S(s)$ do not necessarily need to be symmetric about $\mathbb{R}$.  

\begin{figure}
    \centering
    \includegraphics[width = .3\textwidth]{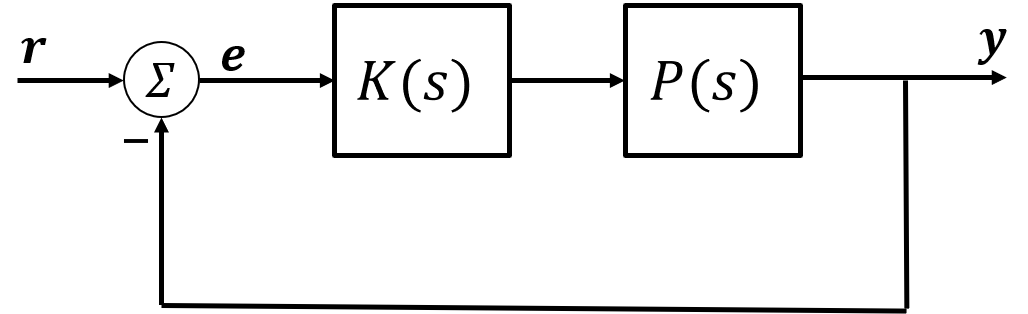} \caption{A simple unity feedback control system.}
    \label{fig:controlSys}
\end{figure}

The following manipulations are classical:
\begin{equation} \label{eq:eq1}
\begin{split}
    I(S) &=\int_{-\infty}^\infty \log|S(iw)|^2 dw\\
    &= \int_{-\infty}^\infty \log(S(iw)S(-iw))dw\\
    &= \int_{-\infty}^\infty (\log(S(iw) + \log(S(-iw)))dw\\
    &= 2\int_{-\infty}^\infty \log(S(iw))dw 
\end{split}
\end{equation}

Let $S(s) = \frac{1}{1 + P(s)K(s)}$, where $P(s)$ and $K(s)$ are the plant and controller transfer functions, respectively, so that 
\begin{equation} \label{eq:I(s)}
\begin{split}
    I(S) &= 2\int_{-\infty}^\infty \log\left(\frac{1}{1+P(iw)K(iw)}\right) dw\\
    &= -2\int_{-\infty}^\infty \log\left(1+P(iw)K(iw)\right) dw\\
    &=  2\int_{-i\infty}^{i\infty} i\log\left(1+P(s)K(s)\right) ds
\end{split}
\end{equation}

Suppose that $P(s) = \frac{N(s)}{D(s)}$ is rational so that $N(s)$ and $D(s)$ are polynomials and $n := \deg(N)$ and $m := \deg(D)$. Furthermore, consider the following \textit{fractal} (or \textit{fractional order}) PID controller, 
\begin{equation}
    K(s) = k_1s^\alpha + \frac{k_{-1}}{s^\beta} + k_0
\end{equation} 
where $k_0$ is the proportional gain, $k_{-1}$ is the integral gain, $k_1$ is the derivative gain. $\alpha$ and $\beta$ represent the derivative and integral orders, respectively, satisfying $0<\alpha<2$ and $0<\beta<2$. We will assume that $m > \alpha + n + 1$ to ensure convergence. Then, 
$1 +P(s)K(s) = 1+ \frac{N(s)(k_1 s^{\alpha+\beta} + k_{-1} + k_0s^\beta)}{D(s)s^\beta}$ so that
\begin{align}
    \log(1 + P(s)K(s)) = \log\left(1+ \frac{N(s)(k_1 s^{\alpha+\beta} + k_{-1} + k_0s^\beta)}{D(s)s^\beta} \right)
\end{align}

\begin{figure}
    \centering
    \includegraphics[width = .3\textwidth]{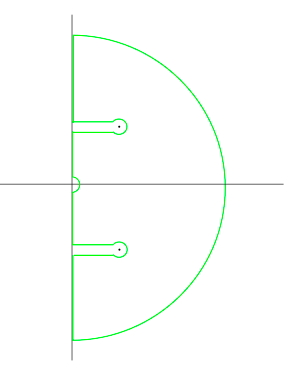}
    \caption{Path for $\gamma$ with branch cuts around poles $p_j$ of $D(s)$.}
    \label{fig:gamma}
\end{figure}

We now prove the following theorem, when $S(s)$ has finitely many poles. 

\begin{theorem} \label{thm:finite}
Suppose $P(s) = \frac{N(s)}{D(s)}$ is rational so that $N(s)$ and $D(s)$ are polynomials, and $n := \deg(N)$ and $m := \deg(D)$, satisfy $m > \alpha + n + 1$. Then
$I(s)$ as defined in \eqref{eq:I(s)} is given by
\begin{equation}
    I(S) = 4 \pi \sum_{p_j}d_j\mathfrak{R}(p_j),
\end{equation}
where $d_j$ is the order of the open loop poles $p_j$ in the right half plane, and the summation is across poles $p_j$. 
\end{theorem} 

\begin{proof}
We consider the path $\gamma$ given in Fig. \ref{fig:gamma} travelling clockwise. Let $\gamma_R$ be the path along the large semi-circle of radius $R$ going clockwise. There are branch cuts at every poles of $D(s)$ in the right half plane. Denote $\gamma_{p_j, \rightarrow}$ to be the portion of the branch-cut moving towards the poles, let $\gamma_{p_j, \leftarrow}$ to be the portion of the branch-cut moving away from the poles, and let $\gamma_{p_j, \epsilon}$ be the portion of the path that goes along a small circle of radius $\epsilon$ centered at $p_j$ going counterclockwise. Finally, let $\gamma_{0, \epsilon}$ be the semicircle of this path of radius epsilon centered at the origin. Since $\gamma$ does not contain poles of $\log(S(s))$, by the calculus of residues, it follows from the Cauchy integral formula that $\int_\gamma\log(S(s))ds = 0$.  We now break the path $\gamma$ as follows:
\begin{align}\label{eq:1}
  0 =  \frac{I(s)}{2i}+ \int_{\gamma_R} \log S(s) ds + \sum_{p_j\neq 0}\int_{\gamma_{p_j, \rightarrow}}\log S(s) ds+
  \sum_{p_j \neq 0}\int_{\gamma_{p_j, \leftarrow}}\log S(s) ds  +\sum_{p_j \ne 0}\int_{\gamma_{p_j,\epsilon}}\log S(s) ds+
  \int_{\gamma_{0, \epsilon}}\log S(s)ds.
\end{align}

The following four lemmas bounding each of the terms in equation \eqref{eq:1} are used to prove the result. Their proofs are deferred to the Appendix.

\begin{lemma}\label{lemma:R}
Under the conditions provided in Theorem \ref{thm:finite} the integral $\int_{\gamma_R} \log S(s) ds$ in equation \eqref{eq:1} satisfies:
\begin{equation}
    \lim_{R \rightarrow \infty}\int_{\gamma_R} \log S(s) ds = 0.
\end{equation}
\end{lemma}

\begin{proof}
See the Appendix. 
\end{proof}

\begin{lemma}\label{lemma:epsilon}
The integral in $\int_{\gamma_{p_j,\epsilon}} \log S(s) ds$ in equation \eqref{eq:1} satisfies:
\begin{equation}
\lim_{\epsilon \rightarrow 0} \int_{\gamma_{p_j,\epsilon}} \log S(s) ds = 0 
\end{equation}
\end{lemma}

\begin{proof}
See the Appendix.
\end{proof}

\begin{lemma}\label{lemma:0}
The integral $\int_{\gamma_{0, \epsilon}}\log S(s) ds$ goes to $0$ as $\epsilon \rightarrow 0$.
\end{lemma}

\begin{proof}
See the Appendix.
\end{proof}

\begin{lemma}\label{lemma:branch}
The $\int_{\gamma_{p_j, \leftarrow}} \log S(s) ds$ and $\int_{\gamma_{p_j, \rightarrow}} \log S(s) ds$ in equation \eqref{eq:1} satisfies:
\begin{align}
  \int_{\gamma_{p_j, \leftarrow}} \log S(s) ds = \int_{\gamma_{p_j, \rightarrow}} \log S(s) ds= \log S(s) ds = 2i \pi \sum_{p_j}d_j\mathfrak{R}(p_j)
\end{align}
\end{lemma}

\begin{proof}
See the Appendix.
\end{proof}

Applying the preceding lemmas to equation \eqref{eq:1} yields:
\begin{align}
    I(S) = -2i\sum_{p_j}\int_{\gamma_{p_j, \rightarrow}} \log S(s) ds- 2i\sum_{p_j}\int_{\gamma_{p_j, \leftarrow}} \log S(s) ds = 4 \pi \sum_{p_j}d_j\mathfrak{R}(p_j)
    \label{eq:final}
\end{align}

Thus proving Theorem \ref{thm:finite}.
\end{proof}

\section{Infinite open right half plane poles}\label{sec:infinite}

In this section, we suppose that $P(s)$ is no longer a rational signal. We can invoke the Weierstrass factorization from \cite{knopp2013theory} to the numerator of sensitivity function to obtain
\begin{equation}\label{eq:weierstrass}
    S(s) = g(s)\prod_{j=1}^\infty (s - p_j)^{d_j},
\end{equation}
where $d_j$ is the order of the open right half plane poles at $s = p_j$ and $g(s)$ is holomorphic on $\mathbb{C}$ with no zeros in the right half plane. Note that since our system is closed loop stable, it has no poles in the right half plane, and thus we do not need to consider them in our contour analysis. Here, $d_j$ doesn't necessarily need to be an integer for our analysis. In the next two sections, we shall show that $I(s)$ as defined in equation \eqref{eq:I(s)} evaluates to the same thing whether the open right half plane poles have limit points in $\mathbb{C}$ or not. 

\subsection{No limit points}

In the case when the open right half plane poles don't form any limit points we have the following result. 

\begin{theorem} \label{thm:nolimit}
 For a sensitivity function $S(s)$, let $I(s)$ be as defined in equation \eqref{eq:I(s)}. Let $p_j$ be the open right half plane poles with order $d_j$ and suppose there are infinitely many of them. Furthermore, suppose that the poles don't form a limit point. Then the bode integral evaluates to:
 
\begin{equation}\label{eq:no_limit}
I(s) =  4 \pi\sum_{p_j} d_j\mathfrak{R}(p_j).
\end{equation}

 where summation across poles $p_j$ are in the right half plane. 
\end{theorem} 

\begin{proof}
To prove this result, we have the following two claims.

\begin{claim}\label{smallest_poles}
There exists an open right half plane pole $p_1 \neq 0$ such that $|p_1| = \inf_p \{|p|:$ $p$ is an open right half plane pole$\}$.
\end{claim}
\begin{proof}
Suppose this wasn't the case, then $\forall \epsilon>0, \exists n(\epsilon)$ such that $|p_{n(\epsilon)}|<\epsilon$. Letting $\epsilon$ go to $0$, it is clear we will get a sequence of poles that tends to $0$. This contradicts the fact that there are no limit points in this sequence of poles. 
\end{proof}

\begin{claim}\label{finite_poles}
For all $R > 0$, the set $\{|p| = R:$ $p$ is an open right half plane pole$\}$ is finite. 
\end{claim}

\begin{proof}
The set $|z| = R$ is a compact set, and thus if $\{|p| = R:$ $p$ is an open right half plane pole$\}$ is infinite, then there must be a converging subsequence, again contradicting there are no limits points. 
\end{proof}

\begin{figure}
    \centering
    \includegraphics[width = .3\textwidth]{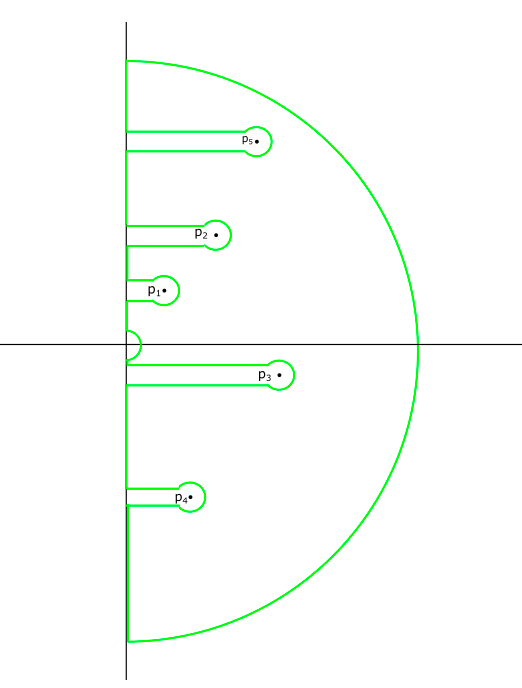}
    \caption{Contour for $\gamma_n$ with $n = 5$.}
    \label{fig:gamma_n}
\end{figure}

By Claim \ref{smallest_poles}, we can order all the open right half plane poles in nondecreasing order of magnitude. Let us now consider the sequence of contours $\{\gamma_n\}_{n=1}^\infty$ with outer radius $R_n = n$ as given in Fig. \ref{fig:gamma_n}. By Claim \ref{finite_poles}, each $\gamma_n$ has finitely many poles inside of the semicircle, which allows us to rigorously apply the method used to prove Theorem \ref{thm:finite} when there are only finitely many poles. For a fixed $R$, if we let the limit as $\epsilon \rightarrow 0$ for each branch cut, then by equation \eqref{eq:1}
\begin{align}\label{eq:limit_general}
  I_{R_n}(S) = 4 \pi\sum_{|p_j|<R} d_j\mathfrak{R}(p_j)-2i\int_{\gamma_R} \log S(s) ds .
\end{align}
We now let $n \to \infty$ (meaning $R_n \to \infty$ as well) to conclude $I(s) = \lim_{n\rightarrow \infty} I_{R_n}(S)$, it follows that 
\begin{equation}
I(s) =  4 \pi\sum_{p_j} d_j\mathfrak{R}(p_j),
\end{equation}
thus proving Theorem \ref{thm:nolimit}.
\end{proof}
\begin{remark}
It's clear to see that if \eqref{eq:no_limit} converges, then the open right half plane poles must approach the imaginary axis. 
\end{remark}
\subsection{Limit points}
Analogous to the last section, in this section we prove the following theorem.
\begin{theorem} \label{thm:limit}
    For a sensitivity function $S(s)$, let $I(s)$ be as defined in equation \eqref{eq:I(s)}. Let $p_j$ be the open right half plane poles of $L(s)$ with order $d_j$ and suppose there are infinitely many of them. Furthermore, suppose there is a single limit points, we have 
    \begin{equation}
        I(s) = 4 \pi \sum_{p_j}d_j\mathfrak{R}(p_j).
    \end{equation}
    where summation across poles $p_j$ are in the right half plane. 
\end{theorem}
\begin{figure}
    \centering
    \includegraphics[width = .3\textwidth]{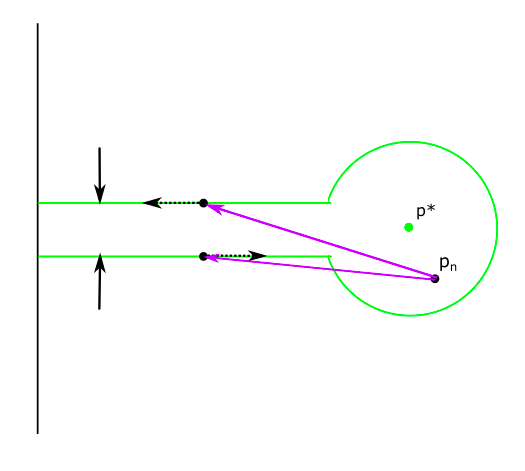}
    \caption{Branch cut at $p^*$ with circle of radius $\epsilon$, with $p_j$ a sequence of poles tending towards $p^*$.}
    \label{fig:branch_limit}
\end{figure}
\begin{proof}
    For simplicity, suppose the right half plane poles have a limit point $p^*$. Let us again consider the sequence of contours $\gamma_n$ with outer radius $R_n$, however, for each $\epsilon >0$ for the poles $p$ such that $|p - p^*| < \epsilon$ we don't put a branch cut for these poles. Instead, we put a branch cut centered at $p^*$ as in Fig. \ref{fig:branch_limit}. By the Weierstrass factorization theorem we have
\begin{align}\label{eq:branch_limit}
     \log S(s) = \log(f(s))+\sum_{|p_j-p^*|<\epsilon}d_j\log(s - p_j) +\log(g(s)),
\end{align}
where $f(s)$ only has poles at greater than $\epsilon$ away from $p^*$. Thus, $f(s)$ has the same phase on the entire branch cut so that $\int_{\gamma_n, \rightarrow }\log(f(s)) ds +  \int_{\gamma_n, \leftarrow }\log(f(s)) ds = 0$. Our goal is now to evaluate $\int_{\gamma_n, \rightarrow }\log(s - p_j) ds + \int_{\gamma_n, \leftarrow }\log(s-p_j) ds$. $s - p_j$ can be represented by the pink arrows in Fig. \ref{fig:branch_limit}. From this, it is clear that as the width of the corridor goes to $0$, if the phase for one of them is $\theta$ then the other is $2\pi - \theta$. Thus, we can parameterize $\gamma_{p_j, \rightarrow}$ in terms of $x$ as $s(x) = r(x)e^{i\theta(x)}$ and $\gamma_{p_j, \leftarrow}$ as $s(x) = r(x)e^{i(\theta(x)+2\pi)}$ where $x \in [0, \mathfrak{R}(p^*) - \epsilon]$ and $r(x) = |s(x)|$. We first evaluate $ \int_{\gamma_n, \rightarrow }\log(s - p_j) ds$ as follows. 
\begin{align}
    \int_{\gamma_n, \rightarrow }\log(s - p_j) ds &=\int_0^{\mathfrak{R}(p_j) - \epsilon}\log\left(r(x)e^{i\theta(x)}\right)\frac{d}{dx}\left[r(x)e^{i\theta(x)}\right]dx\\
    &= \int_0^{\mathfrak{R}(p_j) - \epsilon}\log r(x)\frac{d}{dx}\left[r(x)e^{i\theta(x)}\right]dx+ \int_0^{\mathfrak{R}(p_j) - \epsilon}i\theta(x)\frac{d}{dx}\left[r(x)e^{i\theta(x)}\right]dx. 
\end{align}
We now evaluate 
\begin{align}
     \int_{\gamma_n, \leftarrow }\log(s-p_j) ds &= \int_{\mathfrak{R}(p_j) - \epsilon}^0\log\left(r(x)e^{i(\theta(x)+2\pi)}\right)\frac{d}{dx}\left[r(x)e^{i(\theta(x)+2\pi)}\right] dx\\
     &=\int_{\mathfrak{R}(p_j) - \epsilon}^0 \log r(x)\frac{d}{dx}\left[r(x)e^{i\theta(x)}\right]dx + \int_{\mathfrak{R}(p_j) - \epsilon}^0 i(\theta(x)+2\pi)\frac{d}{dx}\left[r(x)e^{i\theta(x)}\right]dx 
\end{align}
Thus, we have the following result
\begin{align}
       \int_{\gamma_n, \leftrightarrows}\log(s-p_j) ds &= -\int_0^{\mathfrak{R}(p_j) - \epsilon}2\pi i\frac{d}{dx}\left[r(x)e^{i\theta(x)}\right]dx\\
       &= -2\pi i \left(r(\mathfrak{R}(p_j) - \epsilon)e^{i\theta(\mathfrak{R}(p_j) - \epsilon)} - r(0)e^{i\theta(0)} \right).
\end{align}
Thus, using equation \eqref{eq:branch_limit}, for each $\epsilon > 0$ sufficiently small, we obtain
\begin{align}
     \left| \int_{\gamma_n, \leftrightarrows }\log S(s)ds\right| &= \left|\sum_{|p_j-p^*|<\epsilon}  \int_{\gamma_n, \leftrightarrows }\log(s - p_j)ds \right|\\
     &= \left|\sum_{|p_j-p^*|<\epsilon} 2\pi i \left(r(\mathfrak{R}(p_j) - \epsilon)e^{i\theta(\mathfrak{R}(p_j) - \epsilon)} - r(0)e^{i\theta(0)} \right)\right| \\
     &> \sum_{|p_j-p^*|<\epsilon} 2\pi (\mathfrak{R}(p^*)-\epsilon),
\end{align}
where the last inequality comes from the triangle inequality. Thus the sum above is infinite if $\mathfrak{R}(p^*)>0$ since $|\{p_j:|p_j - p^*|<\epsilon$, $p_j$ an open right half plane pole$\}| = \infty$. It follows that any converging subsequence of the open right half plane poles must be converging to a point on the imaginary axis. 
\begin{figure}
    \centering
    \includegraphics[width = .3\textwidth]{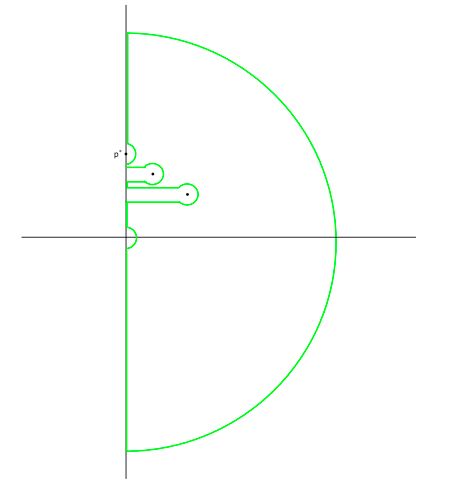}
    \caption{Contour plot for a converging subsequence of poles of $S(s)$ to a point on the imaginary axis. }
    \label{fig:gamma_n,ep}
\end{figure}
Thus, for each limit point $p^*$ we consider $\gamma_R$ as before, but now with a semicircle of radius $\epsilon$ around the limit point $p^*$. See Fig. \ref{fig:gamma_n,ep}. Using the same reasoning as in Lemma \ref{lemma:epsilon}, we conclude that for each value of $R$, $\lim_{\epsilon \rightarrow 0}\int_{\gamma_{p^*, \epsilon}}\log(S(s))ds = 0$ (i.e. the integral of $|S(s)|$ around the circle of radius epsilon centered at $p^*$ goes to $0$ as $\epsilon\rightarrow 0$). Thus, it is  clear that as $\epsilon \rightarrow 0$, the integral of $S(s)$ around this contour is analogous to equation \eqref{eq:final}, that is, it takes on the form:
\begin{equation}\label{eq:no_limit_general}
    I_{R}(s)=4 \pi \sum_{|p_j|<R}d_j\mathfrak{R}(p_j)- \int_{\gamma_R} \log S(s) ds.
\end{equation}

Thus, as $R \rightarrow \infty$, we finish the proof of Theorem \ref{thm:limit} as our desired integral.
\end{proof}
For the sake of simplicity, in equations \eqref{eq:limit_general} and \eqref{eq:no_limit_general} we suppose that the integral going around $\gamma_R$ goes to $0$ as $R \to \infty$. However, under different degree constraints on the sensitivity, if this limit were to give a different value, the Bode integral can be constrained to a smaller value, allowing us to generalize our techniques to a wider class of sensitivity functions. This is formalized by the following corollary.
\begin{corollary}\label{corollary}
For a sensitivity $S$ characterized as in equation \eqref{eq:weierstrass} the Bode integral is 
\begin{equation}
    I(s)=4 \pi \sum_{p_j}d_j\mathfrak{R}(p_j)-
    \lim_{R\to \infty}\int_{\gamma_R} \log S(s) ds.
\end{equation}

where $p_j$ are the right half plane poles. 
\end{corollary}
\section{Conclusion and Future Work}\label{sec:conclusion}
In this paper, we investigated the Bode integral, evaluated what happens for rational systems with fractal PID controllers, and found that the integral is $4\pi$ times the sum of the real parts of the open right half plane poles weighted by the order of the root. We extended our analysis to irrational systems with infinite open right half plane poles and found that when those poles have a limit point it must tend to the imaginary axis. In both cases, the result is the same but replaced with an infinite sum. Thus, the Bode condition can still be used as long as the sum of the real parts of the open right half plane poles converge.

For future work, it would be interesting to investigate the specific conditions on the sensitivity $S$ for the integral in the Bode condition given in Corollary \ref{corollary} to converge. This will allow us to constrain the Bode condition more. In our work, we also suppose the plant denominator is characterized by the Weierstrass factorization theorem, but the case where the plant is characterized by fractional derivatives remains to be investigated. 

\bibliographystyle{plain}
\bibliography{references}

\begin{thebibliography}{10}

\bibitem{aangenent2005nonlinear}
Wouter Aangenent, Ren{\'e} van~de Molengraft, and Maarten Steinbuch.
\newblock Nonlinear control of a linear motion system.
\newblock {\em IFAC Proceedings Volumes}, 38(1):446--451, 2005.

\bibitem{Bode1945}
Hendrik~W Bode.
\newblock {\em Network analysis and feedback amplifier design}.
\newblock D. Van Nostrand, 1945.

\bibitem{chen1997sensitivity}
Jie Chen.
\newblock Sensitivity integrals and transformation techniques: A new
  perspective.
\newblock {\em IEEE transactions on automatic control}, 42(7):1037--1044, 1997.

\bibitem{BodeSensitivity2022}
Yanling Ding, Hui Peng, Junqi Ma, Shinji Hara, and Jie Chen.
\newblock Bode integral: A unifier of control-relevant integral relations.
\newblock {\em IFAC-PapersOnLine}, 55(25):97--102, 2022.
\newblock 10th IFAC Symposium on Robust Control Design ROCOND 2022.

\bibitem{emami2019bodeRevisited}
Abbas Emami-Naeini and Dick de~Roover.
\newblock Bode's sensitivity integral constraints: The waterbed effect
  revisited.
\newblock {\em arXiv preprint arXiv:1902.11302}, 2019.

\bibitem{Freudenberg87}
J.~Freudenberg and D.~Looze.
\newblock A sensitivity tradeoff for plants with time delay.
\newblock {\em IEEE Transactions on Automatic Control}, 32(2):99--104, 1987.

\bibitem{freudenberg1985right}
James Freudenberg and DF~Looze.
\newblock Right half plane poles and zeros and design tradeoffs in feedback
  systems.
\newblock {\em IEEE transactions on automatic control}, 30(6):555--565, 1985.

\bibitem{Freudenberg1988FrequencyDP}
James~S. Freudenberg and Douglas~P. Looze.
\newblock {\em Frequency Domain Properties of Scalar and Multivariable Feedback
  Systems}.
\newblock Springer Berlin, Heidelberg, 1988.

\bibitem{jonckheere1992chaotic}
Edmond~A Jonckheere and Bing-Fei Wu.
\newblock Chaotic disturbance rejection and bode limitation.
\newblock In {\em 1992 American Control Conference}, pages 2227--2231. IEEE,
  1992.

\bibitem{BodeInBiology}
Hiroaki Kitano.
\newblock Towards a theory of biological robustness.
\newblock {\em Molecular Systems Biology}, 3(1):137, 2007.

\bibitem{knopp2013theory}
Konrad Knopp.
\newblock {\em Theory of functions, Parts I and II}.
\newblock Courier Corporation, 2013.

\bibitem{Kwakernaak1972}
Huibert Kwakernaak and Raphael Sivan.
\newblock {\em Linear Optimal Control Systems}.
\newblock John Wiley \& Sons, Inc., USA, 1972.

\bibitem{li2012bode}
Dapeng Li and Naira Hovakimyan.
\newblock Bode-like integral for continuous-time closed-loop systems in the
  presence of limited information.
\newblock {\em IEEE Transactions on Automatic Control}, 58(6):1457--1469, 2012.

\bibitem{Mandelbrot1982}
B.~Mandelbrot.
\newblock {\em The Fractal Geometry of Nature}.
\newblock W.H. Freeman, San Francisco, CA, 1982.

\bibitem{pang2005suppressing}
Chee~Khiang Pang, Daowei Wu, Guoxiao Guo, Tow~Chong Chong, and Youyi Wang.
\newblock Suppressing sensitivity hump in hdd dual-stage servo systems.
\newblock {\em Microsystem Technologies}, 11(8):653--662, 2005.

\bibitem{roy2011studies}
Prateep Roy.
\newblock Studies on the convergence of information theory and control theory.
\newblock In {\em Informatics in Control Automation and Robotics}, pages
  285--299. Springer, 2011.

\bibitem{sariyildiz2021guide}
Emre Sariyildiz.
\newblock A guide to design disturbance observer-based motion control systems
  in discrete-time domain.
\newblock In {\em 2021 IEEE International Conference on Mechatronics (ICM)},
  pages 1--6. IEEE, 2021.

\bibitem{shah2004performance}
Khushboo Shah, Stephan Bohacek, and Edmond Jonckheere.
\newblock On the performance limitation of active queue management (aqm).
\newblock In {\em 2004 43rd IEEE Conference on Decision and Control (CDC)(IEEE
  Cat. No. 04CH37601)}, volume~1, pages 1016--1022. IEEE, 2004.

\bibitem{subramanian2018discrete}
Raaja~Ganapathy Subramanian and Vinodh~Kumar Elumalai.
\newblock Discrete-time setpoint-triggered reset integrator design with
  guaranteed performance and stability.
\newblock {\em ISA transactions}, 81:155--162, 2018.

\bibitem{wan2018sensitivityCDC}
Neng Wan, Dapeng Li, and Naira Hovakimyan.
\newblock Sensitivity analysis of continuous-time systems based on power
  spectral density.
\newblock In {\em 2018 IEEE Conference on Decision and Control (CDC)}, pages
  2549--2554. IEEE, 2018.

\bibitem{wan2019sensitivity}
Neng Wan, Dapeng Li, and Naira Hovakimyan.
\newblock Sensitivity analysis of linear continuous-time feedback systems
  subject to control and measurement noise: An information-theoretic approach.
\newblock {\em Systems \& Control Letters}, 133:104548, 2019.

\bibitem{wan2019simplified}
Neng Wan, Dapeng Li, and Naira Hovakimyan.
\newblock A simplified approach to analyze complementary sensitivity tradeoffs
  in continuous-time and discrete-time systems.
\newblock {\em IEEE Transactions on Automatic Control}, 65(4):1697--1703, 2019.

\bibitem{wan2022simplified}
Neng Wan, Dapeng Li, Lin Song, and Naira Hovakimyan.
\newblock Simplified analysis on filtering sensitivity trade-offs in
  continuous-and discrete-time systems.
\newblock {\em arXiv preprint arXiv:2204.04172}, 2022.

\bibitem{wu1992simplified}
B-F Wu and Edmond~A Jonckheere.
\newblock A simplified approach to bode's theorem for continuous-time and
  discrete-time systems.
\newblock {\em IEEE Transactions on Automatic Control}, 37(11):1797--1802,
  1992.

\bibitem{zhang2002information}
Hui Zhang and Youxian Sun.
\newblock An information theoretic approach to performance limits in linear
  time invariant control systems.
\newblock In {\em 2002 IEEE Region 10 Conference on Computers, Communications,
  Control and Power Engineering. TENCOM'02. Proceedings.}, volume~3, pages
  1424--1427. IEEE, 2002.

\bibitem{zhang2003bode}
Hui Zhang and Youxian Sun.
\newblock Bode integrals and laws of variety in linear control systems.
\newblock In {\em Proceedings of the 2003 American Control Conference, 2003.},
  volume~1, pages 66--70. IEEE, 2003.

\end{thebibliography}
\newpage
\appendix

\section{Proof of Lemmas}

\begin{proof}[Proof of Theorem \ref{lemma:R}]
We have:
\begin{align*}
    \left\vert\int_{\gamma_R} \log S(s)ds\right\vert \leq \int_{\gamma_R} \left\vert\log\left( 1+ \frac{N(s)(k_1 s^{\alpha+\beta} + k_{-1} + k_0s^\beta)}{D(s)s^\beta} \right)\right\vert ds.
\end{align*}
Now letting $N(s) = \sum_{l=0}^n a_ls^l$ and $D(s) = \sum_{l=1}^m b_ls^l$ with $a_l, b_l \in \mathbb{C}$, we get 

\begin{align}
    \left\vert\int_{\gamma_R} \log S(s)ds\right\vert &\leq \int_{\gamma_R} \left\vert\log\left( 1+ \frac{s^{\alpha+n-m}(a_n + \sum_{l=1}^n a_{n-l}s^{-l})(k_1 + k_{-1}s^{-\alpha-\beta} + k_0 s^{-\alpha}}{b_m + \sum_{l=1}^m b_{m-l}s^{-l}} \right)\right\vert ds \\
    &\sim \int_{\gamma_R} \left\vert\log\left( 1+ \frac{a_nk_1}{b_m}s^{\alpha+n-m} \right)\right\vert ds \\
    &= \int_{\gamma_R} \left\vert\sum_{l=1}^\infty \frac{(-1)^{l-1}}{l} \left(\frac{a_nk_1}{b_m} s^{\alpha+n-m}\right)^l\right\vert ds\\
    &\leq \int_{\gamma_R} \sum_{l=1}^\infty  \left(\frac{|a_nk_1|}{|b_m|} |s|^{\alpha+n-m}\right)^l ds \\
    &= \pi R\sum_{l=1}^\infty  \left(\frac{|a_nk_1|}{|b_m|} R^{\alpha+n-m}\right)^l\\
    &=   \frac{\pi \frac{|a_nk_1|}{|b_m|} R^{\alpha+n-m+1}}{1 - \frac{|a_nk_1|}{|b_m|} R^{\alpha+n-m}},
\end{align}

where in the last equality, we used $m > \alpha + n + 1$ to apply the geometric series formula. Thus, the limit $ \left\vert\int_{\gamma_R} \log S(s)ds\right\vert  = 0$.  
\end{proof}
\begin{proof}[Proof of Lemma \ref{lemma:epsilon}]

We have:
\begin{align}\label{eq:eq3} 
    \left\vert\int_{\gamma_{p_j,\epsilon}} \log S(s) ds \right\vert &\leq 2\int_{\gamma_{p_j,\epsilon}} \left\vert\log\left( 1+ \frac{N(s)(k_1 s^{\alpha+\beta} + k_{-1} + k_0s^\beta)}{D(s)s^\beta} \right)\right\vert ds\\
&\leq  M'\int_{\gamma_{p_j,\epsilon}}\log\left( \left\vert 1+ \frac{N(s)(k_1 s^{\alpha+\beta} + k_{-1} + k_0s^\beta)}{D(s)s^\beta}\right\vert \right)ds\\
&\leq M'\int_{\gamma_{p_j,\epsilon}}\log\left(1+ \frac{|N(s)|
    |k_1 s^{\alpha+\beta} + k_{-1}+ k_0s^\beta|}{|D(s)||s|^\beta}\right)ds\\
    &\leq  M\pi \epsilon\log\left(1+\frac{1}{|p_j|^\eta}\right) \to 0
\end{align}

as $\epsilon \to 0$ for some constants $M', M$ and $\eta$. The second inequality follows since for any complex number $z$, we have $|\log(z)| = \sqrt{\log|z|^2 + \theta^2}$ where $\theta = \arg(z) \in [0, 2\pi]$. Thus, when $z$ is large, we have $|\log(z)| \leq C \log|z|$ for some constant $C$. 
\end{proof}

\begin{proof}[Proof of Lemma \ref{lemma:0}]

We have:
\begin{align}
    \left\vert\int_{\gamma_{0,\epsilon}} \log S(s) ds \right\vert &\leq 2\int_{\gamma_{0,\epsilon}} \left\vert\log\left( 1+ \frac{N(s)(k_1 s^{\alpha+\beta} + k_{-1} + k_0s^\beta)}{D(s)s^\beta} \right)\right\vert ds\\
    &\leq  M'\int_{\gamma_{0,\epsilon}}\log\left( \left\vert 1+ \frac{N(s)(k_1 s^{\alpha+\beta} + k_{-1} + k_0s^\beta)}{D(s)s^\beta}\right\vert \right)ds \\
    &\leq M'\int_{\gamma_{0,\epsilon}}\log\left(1+ \frac{|N(s)|
    |k_1 s^{\alpha+\beta} + k_{-1}+ k_0s^\beta|}{|D(s)||s|^\beta}\right)ds\\
    &\leq  M\pi \epsilon\log\left(1+ \frac{1}{\epsilon^\eta}\right)
    \label{eq:3}
\end{align}
for some constant $M$ and $\eta$. The second inequality follows since for any complex number $z$, we have $|\log(z)| = \sqrt{\log|z|^2 + \theta^2}$ where $\theta = \arg(z) \in [0, 2\pi]$. We can evaluate the right hand side of equation \eqref{eq:3} as $\epsilon \rightarrow 0$ using L’H\^{o}pital’s rule:
\begin{align}
    \lim_{\epsilon\rightarrow 0}M\pi \epsilon\log\left(1+ \frac{1}{\epsilon^\eta}\right)=  \lim_{\epsilon\rightarrow 0}M\pi \frac{\log\left(1+ \frac{1}{\epsilon^\eta}\right)}{\frac{1}{\epsilon}}
    = \lim_{\epsilon\rightarrow 0}M\pi \frac{\frac{1}{1+ \frac{1}{\epsilon^\eta}}\frac{-\eta}{\epsilon^{\eta+1}}}{\frac{-1}{\epsilon^2}}
    = \lim_{\epsilon\rightarrow 0}M\pi \frac{\eta}{\epsilon^{\eta - 1} + \frac{1}{\epsilon}} = 0.
\end{align}
\end{proof}

\begin{proof}[Proof of \ref{lemma:branch}]
To evaluate these integrals, we use the fundamental theorem of algebra to write the signal as follows,
\begin{equation}
    S(s) = g(s)\prod_j (s - p_j)^{d_j},
\end{equation}
where $d_j$ is the order of the open right half plane poles of $L(s)$ at $s = p_j$, and where $g(s)$ is holomorphic on $\mathbb{C}$. We have the following relationship for each poles $p_j$ in $D(s)$:
\begin{equation}
    \log S(s) = \log(g(s)) +  d_j\log(s - p_j)  ,
\end{equation}
Thus, $  \int_{\gamma_n, \rightarrow }\log(g(s)) ds +  \int_{\gamma_n, \leftarrow }\log(g(s)) ds = 0$ so we have 
\begin{align}
    \int_{\gamma_n, \rightarrow }\log(S(s)) ds +  \int_{\gamma_n, \leftarrow }\log(S(s)) ds = d_j\int_{\gamma_n, \rightarrow }\log(s - p_j) ds + d_j\int_{\gamma_n, \leftarrow }\log(s-p_j) ds.
\end{align}

\begin{figure}
    \centering
    \includegraphics[width = .3\textwidth]{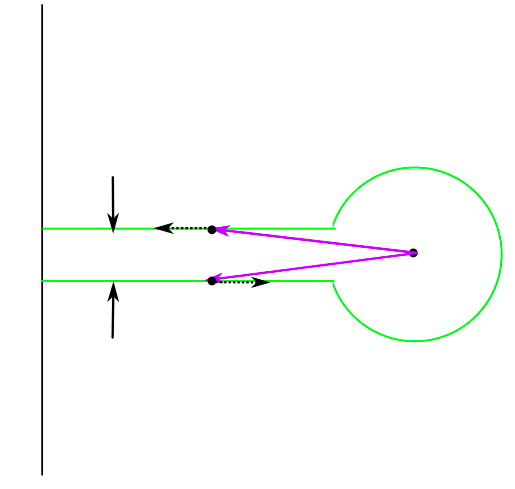}
    \caption{Branch Cut for $\gamma$ at the poles $p_j$.}
    \label{fig:branch}
\end{figure}

We can see from Fig. \ref{fig:branch} that the phase of $s - p_j$ (as given by the pink arrows) goes to $-i\pi$ when $s$ is on $\gamma_{p_j, \rightarrow}$ and $i\pi$ when $s$ is on $\gamma_{p_j, \leftarrow}$ as the width of the corridor goes to $0$. Thus, we can parameterize $\gamma_{p_j, \rightarrow}$ as $s = p_j + xe^{-i\pi}$ and $\gamma_{p_j, \leftarrow}$ as $s = p_j + xe^{i\pi}$. We first evaluate $ \int_{\gamma_n, \rightarrow }\log(s - p_j) ds$ as follows. 
\begin{align}
    \int_{\gamma_n, \rightarrow }\log(s - p_j) ds &=\int_0^{\mathfrak{R}(p_j)}\log\left(xe^{-i\pi}\right)e^{-i\pi}dx\\
    &= -\int_0^{\mathfrak{R}(p_j)}\log(x)dx + \int_0^{\mathfrak{R}(p_j)}i\pi dx\\
    &= \left[ -x\log(x) + x \right]_0^{\mathfrak{R}(p_j)} +i\pi \mathfrak{R}(p_j).
\end{align}

We can now evaluate 
\begin{align}
     \int_{\gamma_n, \leftarrow }\log(s-p_j) ds &= \int_{\mathfrak{R}(p_j)}^0\log\left(xe^{i\pi}\right)e^{i\pi}dx\\
     &= -\int_{\mathfrak{R}(p_j)}^0\log(x)dx -  \int_{\mathfrak{R}(p_j)}^0 i\pi dx\\
&= \left[ -x\log(x) + x \right]_{\mathfrak{R}(p_j)}^0 + i\pi \mathfrak{R}(p_j),
\end{align}
We can add the integral around $\gamma_n, \rightarrow$ and $\gamma_n, \leftarrow$ together to get the desired result. 
\end{proof}

\end{document}